\documentclass[dvipsnames,format=sigconf]{acmart}
\usepackage{graphicx}
\usepackage{latexsym}
\usepackage{amsmath}
\usepackage{algorithm,algorithmic, tabularx}
\usepackage{stackengine}
\usepackage[flushleft]{threeparttable} 
\usepackage{color}

\theoremstyle{definition}
\usepackage{multirow}

\theoremstyle{remark}

\newif\ifcom
\newtheorem{theorem}{Theorem}

\newtheorem{lemma}[theorem]{Lemma}

\usepackage{comment}

\AtBeginDocument{%
  \providecommand\BibTeX{{%
    \normalfont B\kern-0.5em{\scshape i\kern-0.25em b}\kern-0.8em\TeX}}}
\usepackage{lipsum}

\newcommand\blfootnote[1]{%
  \begingroup
  \renewcommand\thefootnote{}\footnote{#1}%
  \addtocounter{footnote}{-1}%
  \endgroup
}

\begin{document}

\title{Optimizing Cyber Response Time on Temporal Active Directory Networks Using Decoys (Extended Version)}


\author{Huy Q. Ngo}
\affiliation{%
  \institution{The University of Adelaide}
  \streetaddress{Adelaide, Australia}
  \city{Adelaide}
  \country{Australia}}
\email{quanghuy.ngo@adelaide.edu.au}

\author{Mingyu Guo}
\affiliation{%
  \institution{The University of Adelaide}
  \streetaddress{Adelaide, Australia}
  \city{Adelaide}
  \country{Australia}}
\email{mingyu.guo@adelaide.edu.au}

\author{Hung X. Nguyen}
\affiliation{%
  \institution{The University of Adelaide}
  \streetaddress{Adelaide, Australia}
  \city{Adelaide}
  \country{Australia}}
\email{hung.nguyen@adelaide.edu.au}

\renewcommand{\shortauthors}{Ngo, et al.}

\begin{abstract}
Microsoft Active Directory (AD) is the default security management system for Window domain network. We study the problem of placing decoys in AD network to detect potential attacks. We model the problem as a Stackelberg game between an attacker and a defender on AD attack graphs where the defender employs a set of decoys to detect the attacker on their way to Domain Admin (DA). Contrary to previous works, we consider time-varying (temporal) attack graphs. We proposed a novel metric called response time, to measure the effectiveness of our decoy placement in temporal attack graphs. Response time is defined as the duration from the moment attackers trigger the first decoy to when they compromise the DA. Our goal is to maximize the defender's response time to the worst-case attack paths. We establish the NP-hard nature of the defender's optimization problem, leading us to develop Evolutionary Diversity Optimization (EDO) algorithms. EDO algorithms identify diverse sets of high-quality solutions for the optimization problem. Despite the polynomial nature of the fitness function, it proves experimentally slow for larger graphs. To enhance scalability, we proposed an algorithm that exploits the static nature of AD infrastructure in the temporal setting. Then, we introduce tailored repair operations, ensuring the convergence to better results while maintaining scalability for larger graphs.
\end{abstract}




\maketitle

\section{Introduction} 
\blfootnote{The manuscript have been accepted as full paper at The Genetic and Evolutionary Computation Conference (GECCO) 2024}

Active Directory is Microsoft's identity and access management system designed for Windows domain networks. It's widely adopted and plays a critical role in the networks of many enterprises and government bodies. However, its popularity has also made it a prime target for many cyber adversaries over the years. According to a report from Microsoft ~\cite{stat2}, there has been an alarming surge in attacks targeting AD users, with 30 billion attempted password attacks on AD accounts reported each month in 2023.

An Active Directory network naturally describes an attack graph, with \textbf{nodes} representing both physical and virtual entities such as users, computers, security groups, etc., and \textbf{directed edge} $(i, j)$ representing the vulnerability and accesses that an attacker can exploit to gain access from node $i$ to node $j$. BloodHound \footnote{https://github.com/BloodHoundAD/BloodHound} is one of the most influential tools for analysing/visualising the AD attack graph. BloodHound models the \textit{identity snowball attack}, a concept initially introduced by Dunagan et al.~\cite{dunagan2009heat}. The identity snowball attack models the sequence of attack in the network allowing them to gain access of higher privileges nodes from a low privilege node (ex. Account A $\xrightarrow[]{\text{AdminTo}}$ Computer B $\xrightarrow[]{\text{HasSession}}$ Account C).  However, Dunagan et al. \cite{dunagan2009heat} and several other works on defending Active Directory network \cite{guo2022practical, guo2023scalable, goel2022defending, goel2023evolving, zhang2023oracle} over-simply the problem with the assumption that AD network is static. 
In practice, the AD graph is very dynamic which will effect the security landscape overtime. One of the major sources of changes in the AD graphs is caused by users’ activities. In Windows systems, user authentication leaves behind credential material, typically in the form of a hash or clear-text password in the computer's memory. Adversaries can exploit this vulnerability, harvesting credentials for lateral movement. In the BloodHound, this vulnerability is presented as 'HasSession'. HasSession edges will stay online until being removed from the graph when the user signs off from the computer after a period of time. In this work, we formally model the dynamics of the AD graph using the \textbf{temporal attack graph}, wherein attackers gain access to nodes in the AD graph through the \textit{identity snowball attack}, presented as \textit{temporal paths}. For example, in Figure \ref{fig:ADex}, the identity snowball attack in temporal graph for gaining access of account $U_3$ from compromised node $s_2$ can be the following temporal path: $s_2$ $\xrightarrow[]{\text{1}}$ $Cp_1$ $\xrightarrow[]{\text{2}}$ $Cp_3$ $\xrightarrow[]{\text{4}}$ $U_3$ where number on each arrow is time the attacker exploit the edge to gain the access to next node. The static attack graph can not capture this attack path if generated at time steps $t \in [1,4) \cup (6, 10]$

\begin{figure}[h]
  \includegraphics[width=1\linewidth]{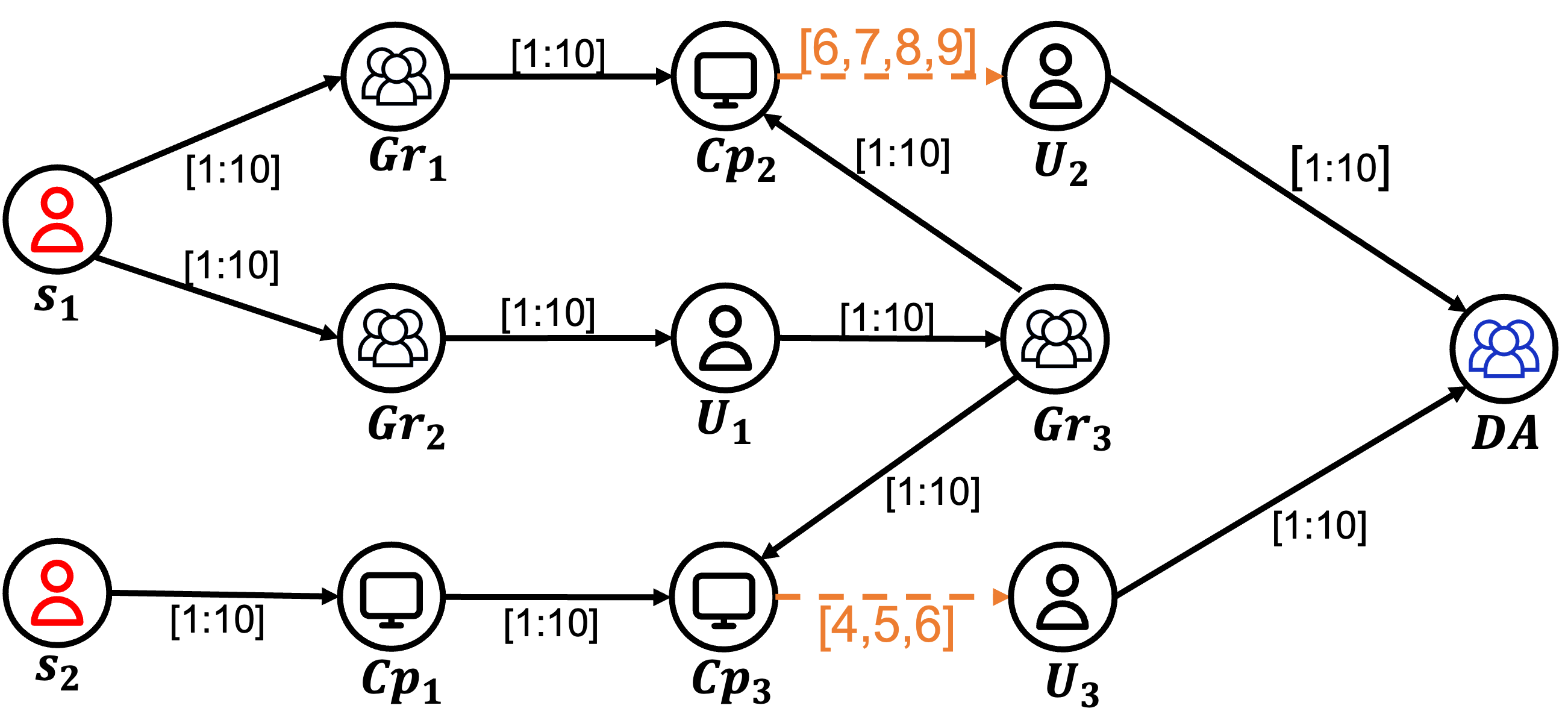}
  \caption{Example of an Active Directory graph sampled over a period of 10 time units. The timestamps on each edge indicates its appearance time. Black labels represent static edges, while orange labels denote dynamic edges (HasSession).}
  \label{fig:ADex}
\end{figure}

In this paper, we study a method for defending temporal AD attack graph by using active defense with cyber decoys. Decoys or honeypots \cite{decoy1, decoy2} are fake assets such as fake users, and fake hosts that trigger an alert when attackers engage them. They are designed to lure attackers to them by mirroring authentic assets. By allocating decoys in strategic locations, they can serve as the sentinel for the early detection of the cyber threat. An early detection of a threat can increase the effectiveness of incident response process of IT admin and reduce the further damage caused by the attack~\cite{stat1}. Motivated by the early detection use case of the network decoy, we proposed a problem for decoy allocation in temporal network called $maxRT$. In this problem, defender aims to optimize the \textbf{response time} of the decoy to any attack path. The response time is defined as the duration from the moment attacker triggers the first decoy to when they can reach DA. Defender want to maximize the response time to ensure the early detection, providing IT admin with sufficient time to respond to the incident before the attacker reaches the DA. 

We model our problem as a two-player Stackelberg game model with a pure strategy. In our game, the defender (leader) wants to allocate at most $b$ decoys on a set of \textit{blockable} nodes. The defender's allocation intercepts the attacker's temporal attack path while maximizing the response time of the allocation to ensure early detection. The attacker (follower) has access to a set of compromised entry nodes. The attacker can also observe the entire temporal AD graph and the defensive strategy. This assumption is based on the practicality of attackers employing reconnaissance tools similar to Sharphound \footnote{https://github.com/BloodHoundAD/SharpHound} to collect data from domain controllers and build an AD attack graph. The attacker’s strategy specifies an entry node, and from it, a temporal attack path to DA. The attacker's best response is to choose a temporal attack path that has the minimum response time. We will later show that the attacker's optimal plan can be found in polynomial time.

\textbf{Our Contribution.} This paper aims to propose a solution for the decoys allocation in Active Directory problem. We first prove $\mathcal{NP}$-hard nature of the defender's combinatorial optimization problem. We then introduce the Evolutionary Diversity algorithm as a heuristic solver. However, the vanilla EDO algorithm does not scale well for our problem when it fails to converge to any feasible solution (response time > 0) in some graphs. When we mention "vanilla" EDO algorithm, we refer to directly applying the EDO implementation from Goel et al. \cite{goel2022defending, goel2023evolving} to our problem. In an attempt to run the vanilla EDO algorithm on the ADX10 graph in our experiment, the response time of the solution remains 0 even after 2 million iterations (equivalent to almost 3 days of computational effort).
To enhance our algorithm, we propose several improvements. Firstly, the computation of the attacker's optimal path relies on the earliest-arrival path, which is computationally slow in AD graphs. We observe an contradiction that the state-of-the-art algorithm for computing the earliest-arrival path becomes highly inefficient in temporal graphs with a large number of static edges. Despite the dynamic characteristics of the AD graph, a significant portion of the AD infrastructure remains static. To address this problem in AD-specific graphs, we present a novel Dijkstra-based algorithm for computing the earliest-arrival path which significantly improves the run-time of the fitness function.
Secondly, we introduce two constraint-handling techniques to tackle the difficulty of finding feasible solutions in the vanilla EDO. The first method introduces a repair mechanism using Integer Linear Program (ILP) to directly patch the infeasible solution every round. The second approach introduces the surrogate/penalty fitness function. The surrogate function is a lightweight fitness function that replaces the computationally expensive real fitness function, allowing the evaluation of individuals at a lower cost during each iteration. The surrogate function is designed to evaluate a solution on a set of "important" attack paths instead of the whole graph and penalize the infeasible individuals. We experimentally verify that our proposal effectively improves the scalability of the EDO algorithms on our problem.

\section{Model Description}


\subsection{Background}
\textbf{Temporal directed graph} define as $G = (V, E_1,\cdots, E_{t_{max}}) = (V, E = (E_i)_{i\in [t_{max}]})$) where $V$ is a set of vertices in the graph and $E_i$ is the set of edges at time $i$. We denote the tuple $(u, v, t) \in E_t$ the edge from $u$ to $v$ appears at time $t$. For the sake of model simplicity, we assume that every edge has a duration of 1. In other words, if an attacker traverse edge $(u, v, t)$, they will start from node $u$ at time $t$ and arrive $v$ at time $t + 1$. Despite this simplification, all our algorithms remain effective in more general settings in which the duration of every edge is larger or equal 1. We call $t_{max}$ the lifetime of the graph. 
We also define the \textbf{underlying graph} of graph $G$ as $G_{\downarrow} = (V, E_{\downarrow})$ where $E_{\downarrow} = \bigcup_{t=1}^{t_{max}}E_t$. For the ease of demonstration in the paper, we also denote $time\_label(u, v) = (t_i)_{i=1}^k$ as a (ascending) sorted list of time units that edge $(u, v)$ appears (or is on) if $(u, v, t_i) \in E_i$ where $t_{i} \in time\_label(u, v)$. Otherwise, we say edge $(u, v)$  disappears or is off at time step $t_i$ if $t_i \notin time\_label(u, v)$. We denote a set of \textbf{static edges} as $E_s$, we say an edge $(u, v) \in E_s$ is static if they appear in every time step throughout the graph's lifetime or $|time\_label(u, v)| = t_{max}$. Similarly, we denote set of \textbf{dynamic edges} as $E_d$, we say an edge $(u, v) \in E_d$ is dynamic if they disappear from the graph for some of the time units or $|time\_label(u, v)| < t_{max}$.

\textbf{Temporal $(s, d)$-path} is defined as a sequence of edges in graph $G$ exhibiting a monotonic increase in edge labels. For any two distinct nodes $s, d \in V$, a temporal path between two vertices $s$ and $d$ is represented by the sequence of edges: $\pi$ = $\pi(s, d) = \langle(s = v_0, v_1, t_1), (v_1, v_2, t_2), \dots, (v_{k-1}, v_k = d, t_k) \rangle = \langle(v_{i-1}, v_i, t_i )\rangle_{i=1}^{k}$ where $v_i \neq v_j$ and $t_i < t_j$ for all $i, j \in \{0, \dots, k\}$ with $i \neq j$.  
We denote $start(\pi) = t_1$ and $end(\pi) = t_k + 1$ as the \textbf{starting time} and \textbf{ending time} of a path $\pi(s, d)$. We further denote by $dur(\pi(s, d)) = end(\pi(s, d)) - start(\pi(s, d))$ the duration of travelling from the starting vertex to the ending vertex of the path $\pi(s, d)$. 
Next, we define a set of every possible temporal path from $s$ to $d$ between interval $[t_{\alpha}, t_{\omega}]$ as $\Pi(s, d, [t_{\alpha}, t_{\omega}]) = \{\pi:\pi \text{ is a (s, d)-temporal path such} $ $ \text{that} start(\pi) \geq t_{\alpha}, end(\pi) \leq t_{\omega} \}$. Then, a path $p \in \Pi(s, d, [t_{\alpha}, t_{\omega}])$ is an \textbf{earliest-arrival path} if $end(\pi) = min\{ end(\pi'): \pi' \in \Pi(s, d, [t_{\alpha}, t_{\omega}])\}$.


\textbf{Temporal $(s, d)$-cut}, also known as a temporal $(s, d)$-separator, refers to the set of nodes $C(s, d)$ in the graph $G$ that the removal of every node in set $C(s, d)$ will disconnects all temporal paths from $s$ to $d$. \textit{It is essential to note that in this paper, the terms "cut" or "separator" specifically refers to the allocation of decoy on vertices}. When we employ a temporal $(s, d)$-cut the graph, we guarantee every (s, d) temporal path has a contact with the cut $C$.
Given a path $\pi(v_0, v_k) = \langle(v_{i-1}, v_i, t_i)\rangle_{i=1}^{k}$ in graph $G$ and a cut $C(v_0, v_k)$, we define a node $v$ as the \textbf{first point of contact} between the path $\pi(v_0, v_k)$ and the cut $C(v_0, v_k)$ if $v \in I : \forall u \in I, \text{dur}(v_0, v) \leq \text{dur}(v_0, u)$ where $I = V(\pi(v_0, v_k)) \cap C(v_0, v_k)$. In plain language, the first point of contact represents the first honeypot encountered when following the path. 

\textbf{Response time} denoted as $RT$ is a key parameter introduced in this paper for our specific problem. 
The response time of a path $\pi$ is defined as the duration between the moment when the attacker encounters or triggers the first honeypot and the time when the attacker compromises the Domain Admin while following path $\pi$.  
Let's us consider the temporal path $\pi(s, DA) = \langle(s = v_0, v_1, t_1), \dots, (v_{k-1}, v_k = DA, t_k) \rangle = \langle(v_{i-1}, v_i, t_i )\rangle_{i=1}^{k}$ with $v_x$ where $1 < x < k$ is the first point of contact of $\pi(s, DA)$ and the defense solution $C(s, DA)$. The response time of path $\pi(s, DA)$ is defined as $RT(\pi, C) = dur(\pi(s, DA)) - dur(\pi(s, v_x)) = t_k - t_{x} $
As a defender, we want to maximize the response time of every temporal path in the attack graph to let IT admin have enough time to react to the incident. 

\textit{Example 2.1.} Figure \ref{fig:ADex} illustrates a temporal Active Directory graph. The graph includes two compromised users, denoted as $s_1$ and $s_2$. The graph consists of two sets of edges: static edges, allowing the attacker to move between nodes at every time step, and dynamic edges, which appear for a limited time. In this example, we assume the defender allocates honeypots to a set of nodes $C = \{Cp_2, Cp_3\}$.
Consider the following temporal path $\pi = \langle(s_1, Gr_1, 2), (Gr_1, Cp_2, 4),$ $ (Cp_2, U_2, 6), (U_2, DA, 7) \rangle$ from $S$ to $DA$. Assuming the attacker from $s_1$ chooses this path, the honeypot on node $Cp_2$ is triggered at time 4 (as the attacker steps on it), alerting the IT admin to the attacker's presence. In this context, node $Cp_2$ is considered as the first point of contact for the attacker. The response time, defined as the time from honeypot alert to the attacker compromising the DA, is $RT = dur(\pi(s_1, DA)) - dur(\pi(s_1, Cp_2)) = (7+1-2) - (4+1-2) = 3$ units (plus 1 due to the assumption that traversing each edge takes 1 time unit). During this window, the IT admin has 3 time units to isolate compromised systems and terminate the attacker's unauthorized session. Note that the proposed response time is a realistic model of real hackers' behaviour where they would wait in the system for a long time before an opportunity arises for the next movement.

\subsection{Problem formulation}

\textbf{Temporal directed attack graph}
We define an AD attack graph in our model as a Temporal directed graph $G = (V, E_1,\cdots, E_{t_{max}})$. Set of vertices V represents all physical and virtual entities such as user, computer, security group, etc. The set of edge $E_i$ denotes the link modelling the security dependency and relationships between entities which represent vulnerabilities for attacker to make lateral movements.
There is a set $S\subseteq V$ of initial footholds called entry vertices, and the attacker has already compromised these vertices at the start of the attack. 
The attack goal is to compromise the Domain Admin (DA), the attacker can laterally move through the network using any of the \textbf{temporal (s, DA)-path}.

\textbf{Formulation with game theory} The problem of defending a temporal AD network with honeypots can be modelled as a Stackelberg game.
In our proposed model, the defender can deploy a set of honeypots on a set of vertices $C$ (a cut) such that form a temporal $(S, DA)$-cut. In our model, each honeypot will "monitor" any malicious activities on their allocated vertices. The honeypots will set an alert to IT admin once the attacker steps on one of these vertices. Defender can only allocate honeypot on a set of blockable vertices, denoted by $N_b \subseteq V$. 
In consideration of a worst-case scenario, we assume the attacker has full visibility into the temporal graph and the honeypot placements. The attacker can bypass these honeypots if the honeypot's placement does not form a temporal $(s, DA)$-cut, which in this case, the response time is 0. 
Consequently, when the budget of the defender problem is exactly the size of the minimum temporal cut, our problem's solution is also the solution for the minimum temporal $(s, d)$-separator problem \cite{zschoche2020complexity} which is known to be a $\mathcal{NP}$-complete problem. However, our problem goes beyond this by also maximizing the response time of the temporal cut which tends to "push" the solution further away from the DA. Generally speaking, nodes further away from the DA tend to be lower privilege nodes instead of servers or admin. Therefore, our solution incurs lesser disruption to the network.
We say $C$ is a defender's \textbf{feasible solution} if $C$ is strictly a temporal $(S, DA)$-cut, otherwise, it is a \textbf{infeasible solution}. Strategically, when facing a defence solution $C$, the attacker selects a path that minimizes the response time. The \textbf{attacker optimal attack} path can be found via $\min_{\pi\in \Pi} RT(\pi, C)$ where $\Pi$ is the set of every possible temporal path between each vertex $s \in S$ to DA. In contrast, the defender aims to find a cut $C$ that maximize the response time. The \textbf{defender's objective} is formulated as

\begin{equation}
\max_{C\subseteq N_{b}, |C|<b}\{\min_{\pi\in \Pi} RT(\pi, C)\}.
\end{equation}

\begin{theorem} Defender's problem is $\mathcal{NP}$-hard.
\label{theorem:np}
\end{theorem}

\textit{Due to space constraints and anonymity, we omit the proofs of every theorems in this submission. Detailed proofs will be provided in the extended technical report, supplementing the main manuscript.}

\section{Related Work}

\textbf{Identity Snowball Attack in dynamic environment}. In the literature, there are several efforts to model the identity snowball attacks with consideration of the dynamic nature of the attack graph. Ngo et al. \cite{ngo2023catch} also study the honeypot/decoys allocation on Active Directory network with the consideration of the dynamic setting. However, their approach to modelling dynamism is somewhat simplistic. They capture the dynamic nature by taking independent static snapshots of the attack graph at each time step, treating each snapshot as an attacker's scenario in a static graph. Their allocation strategy jointly optimizes the number of attack paths in each snapshot. Ngo et al. \cite{ngo2023catch} fail to model the identity snowball attack in the temporal graph. In practical scenarios, attackers can patiently "lurk" in a node until a more opportune path emerges. This characteristic makes our model more sophisticated and practical than theirs. Albanese et al. \cite{albanese2022formal} attempted to model the credential hopping attacks/identity snowball attacks on the time-varying user-computer graph. They assume that attacker does no observation on the network topology and employ a heuristic algorithm to find the upper-bound of the attacker attack effort. In contrast, our work considers the worst scenario where attacker have the observation on the attack graph and we can derive the optimal attack response. Pope et al. \cite{pope2018automated} also consider a similar model to Albanse et al. except they employ genetic programming to predict the attacker success rate/effort. We highlight that none of these works considers the temporal graph for modelling the dynamic of AD graph. 

\textbf{Active Directory.} In the literature, two primary defender strategies have been explored for defending Active Directory: edge-blocking and decoy allocation (node-blocking). The seminal work by Dunagan et al. \cite{dunagan2009heat} was the first to study the defense problem in Active Directory through edge-blocking by introducing the heuristic edge-blocking algorithm. Follow-up researches on the edge-blocking optimization problem includes Guo et al. \cite{guo2022practical} proposed an optimal edge-blocking strategy using Fixed-Parameter Tractable algorithms; \cite{guo2023scalable, zhang2023oracle} improved scalability through Mixed-Integer Programming and the Double Oracle algorithm; Goel et al. \cite{goel2022defending, goel2023evolving} proposed the Evolutionary Diversity Optimization (EDO) algorithm to defend against attackers in a configurable environment; and Guo et al.\cite{guo2023limited} studied optimal edge-blocking problem with minimal human input. Another approach for defending Active Directory found in the literature involves node-blocking, which abstracts the concept to decoy allocation. Ngo et al. \cite{ngo2023catch} are the first to study the honeypot allocation problem for defending Active directory where they proposed MIP algorithm to solve the problem.

\textbf{Evolutionary Diversity Optimization} \cite{ulrich2010integrating} is a recent branch of Evolutionary Computation. EDO is designed to identify a set of solutions that is both high-quality and structurally diverse. In the literature, there have been considerable efforts exploring the EDO algorithm for various combinatorial problems, including the travelling salesperson problem \cite{nikfarjam2021entropy, do2022analysis, bossek2019evolving}, minimum spanning tree problem \cite{bossek2021evolutionary}, knapsack problems \cite{bossek2021breeding}, and more. Among these studies, the work of Goel et al. \cite{goel2022defending, goel2023evolving} is particularly relevant to our research. Goel et al. consider the edge-blocking problem against attacker in AD graph where edges are associated with a failure rate and detection rate. They deploy a neural network/reinforcement learning to approximate the attacker's strategy and apply EDO algorithm to solve the defender problem. In our study, our EDO algorithm draws inspiration from Goel et al. \cite{goel2022defending}, including the design of the mutation/crossover operator and diversity measure strategy. However, our experimental findings reveal that the vanilla EDO algorithm performs poorly when directly applied to our specific problem.

\section{Proposed methodology}

\subsection{Game-theoretical rational attacker}
\label{sec:optatk}
In our model, the game-theoretical rational/optimal attacker will choose the attack path that has the minimal response time. We illustrate such paths using the following example from Figure \ref{fig:ADex}. We assume the defender allocates honeypots to a set of nodes $C = {Cp_2, Cp_3}$. 
Starting from the entry node $s_1$, let's examine two potential attack paths: $\pi_1 = \langle(s_1, Gr_1, 1), (Gr_1, Cp_2, 2),(Cp_2, U_2, 6), (U_2, DA, 7) \rangle$ and $\pi_2 = \langle(s_1, Gr_1, 1), (Gr_1, Cp_2, 5),(Cp_2, U_2, 6), (U_2, DA, 7) \rangle$. The difference between these 2 paths lies in the departure time of exploiting the second edge $(Gr_1, Cp_2)$. After exploiting the first edge $(s_1, Gr_1)$ at time 1, the attacker has 2 options: either immediately exploit the next edge at time 2 ($\pi_1$) or wait until time 5 to continue ($\pi_2$). Despite both paths leading to the attacker reaching DA at time 7, the attacker is more "troublesome" if they opt for $\pi_1$. This is because the decoy only identifies them at time 5 ($RT = 2$) for path $\pi_1$, whereas for path $\pi_2$, the attacker is detected at time 2 ($RT = 5$), providing the defender with significantly more time to respond to the incident. $\pi_1$ in this example is actually the worst-case/optimal attack path.


Algorithm \ref{alg:optatk} for finding such paths can be described as follows. Let's consider an attack graph $G$ and a defender's honeypot allocation $C \in V$. We define a tuple $(\pi_1, c, t_c)$, where $\pi_1$ represents a temporal path, $c$ is a node in $C$, and $t_c$ is a time. Firstly, for each node $c \in C$, we verify if it is reachable from any of the entry nodes $s\in S$ at time $t_c$ in a graph $G' = ((V\setminus C)\cup c, E)$ (line 4) —here, we remove all nodes in $C$ except node $c$ (line 2). The condition in line 4 ensures the \textit{worst-case} condition of the optimal attack path. If we can reach node $c$ from $S$ at time $t_c$ using path $\pi_1$, we then find the earliest-arrival path $\pi_2$ from $c$ to DA within the interval $[t_c, t_{\omega}]$ (line 5-6). We add the tuple $(\pi_1, \pi_2)$ to $\Psi$ (line 7). Next, for every tuple $(\pi_1, \pi_2) \in \Psi$, we merge 2 path to form a temporal $(s,DA)$-path $\pi = \pi_1 + \pi_2$. We identify the tuple with the smallest duration $dur(\pi_2)$, \textit{the duration of the earliest-arrival path $\pi_2$ is actually the response time for the attack path $\pi = \pi_1 + \pi_2$} (line 8). Therefore, the optimal attack path $\pi_{OPT}$ is the one where the $\pi_2$ sub-path has the smallest duration. The fitness function giving a defender solution $C$ can be defined as: 

\begin{equation}
\label{equa:fitness}
  f(C)=\begin{cases}
    \min_{\pi \in \Pi} RT(\pi, C), & \text{if $C$ is feasible (temporal cut)}.\\
    0, & \text{otherwise}.
  \end{cases}
\end{equation}

\begin{algorithm}[H]
 \caption{Algorithm for Computing Optimal Attack Path}
 \label{alg:optatk}
 \begin{algorithmic}[1]
 \renewcommand{\algorithmicrequire}{\textbf{Input:}}
 \REQUIRE Temporal graph $G$, set of source nodes $S$, set of honeypot $C$, destination node DA, time interval $[t_{\alpha}, t_{\omega}]$\\
 \renewcommand{\algorithmicrequire}{\textbf{Output:}}
 \REQUIRE Optimal attack path $\pi$\\
 \STATE \textbf{foreach} $c \in C$ \textbf{do}
 \STATE \quad Remove nodeset $C \setminus c$ from graph $G$
 \STATE \quad \textbf{foreach} $t \in [t_{\alpha}, t_{\omega}]$ \textbf{do}
 \STATE \quad \quad \textbf{if} $c$ can be reached from any $s\in S$ at time $t$ \textbf{do}
 \STATE \quad \quad \quad Store the path used to reach c by time $t$ to $\pi_1$ 
 \STATE \quad \quad \quad $\pi_2 \leftarrow compute\_earliest\_arrvl\_path(G, c, DA, [t, t_{\omega}])$
 \STATE \quad \quad \quad Add $(\pi_1, \pi_2)$ to $\Psi$
 \STATE $\pi = \arg\min_{(\pi_1, \pi_2) \in \Psi} \text{dur}(\pi_2)$
 \STATE \textbf{return} $\pi$
 \end{algorithmic}
\end{algorithm}

For computing earliest-arrival path subroutine (line 6) we can use the state-of-the-art algorithm proposed by Wu et al. \cite{wu2014path} which has been proven to be time-polynomial. This makes computing attacker optimal attack plan time polynomial. Despite this, Wu's algorithm is inefficient when running on AD-specific graph, slowing down the computation of the optimal attack plan. We will discuss this issue in the next section and propose a more efficient approach for calculating the earliest-arrival path.

\subsection{Faster computation for earliest-arrival path}

As outlined in Section \ref{sec:optatk}, the response time of an attack path is determined by the duration of the earliest-arrival path from an initial point of contact to the DA. Therefore, the computation of optimal attack for fitness function required the call of computing the earliest arrival path subroutine. The first candidate algorithm that we use for computing the earliest arrival path in our implementation is Wu's algorithm \cite{wu2014path}. In \cite{wu2014path}, the author explored the computation of minimal paths in temporal graphs, including the earliest-arrival path. Wu et al. introduced a one-pass algorithm for computing the earliest-arrival path, which stands as one of the state-of-the-art algorithms for this task. Wu's algorithm generates a set of edge streams, a chronological sequence of all edges $E$ ordered by the time at which the edge is collected. The algorithm scans through the edge stream, greedily updating the earliest arrival time at each node that satisfies the arrival condition. This process required the duplication of every static edge to correctly update the earliest arrival time which explain the contradictory of the inefficiencies of Wu's algorithm in AD-specific temporal graph. In general, Wu's algorithm poses inefficiencies when applied to graphs with a substantial number of static edges as Wu's algorithm requires the scan of every edge in $E$. In practice, while the AD graph exhibits dynamic characteristics, a significant portion of the AD infrastructure remains static. For instance, in a snapshot taken from the University of Anonymous on 13/10/2021 at 02:00 pm, a total of 1,151,962 relationships (edges) were identified as online at that time while only 4,039 of these edges were the HasSession edges, which are deemed as the primary source contributing to the dynamism of the AD graph.

Our proposed approach utilises the Dijkstra's edge scanning strategy which allows us to perform the scan only on the underlying edges $E_{\downarrow}$. The intuition behind this algorithm lies in using the Dijkstra greedy scanning strategy, which scans through each underlying edge only once to expand the earliest-arrival paths. The pseudocode is given in Algorithm \ref{alg:dijkstraea}. The idea of using Dijkstra for finding earliest-arrival path has been proposed in \cite{xuan2003computing}. However, we further enhance the runtime on graphs with numerous static edges by introducing a conditional statement between lines 11-14 in Algorithm \ref{alg:dijkstraea}

The correctness of the Dijkstra Greedy Strategy for computing the earliest-arrival path is provided in Theorem \ref{theo:dijkstra}. In the general case, the time complexity of Wu et al.'s algorithm can be expressed in our notation as $\mathcal{O}((\varepsilon_{s} + \varepsilon_{d}) \cdot t_{max})$, whereas the time complexity of our proposed algorithm is $\mathcal{O}((\varepsilon_{s} + \varepsilon_{d} \cdot t_{max})\log{}(|V|))$. In scenarios where the number of static edges $\varepsilon_{s} = |E_s|$ outweighs the number of dynamic edges $\varepsilon_{d} = |E_d|$, our algorithm demonstrates more efficient runtime, as theoretically presented in Theorem \ref{theo:staticea}.

Experimentally, when we use these algorithms to find the earliest path from every source to every node in graph $ADX10$ (section ), while Wu's algorithm takes $18.370$ seconds to complete the task, Dijkstra Greedy's runtime is only about $3.389$ seconds (5x faster).  

\begin{theorem}
\label{theo:dijkstra}
Algorithm \ref{alg:dijkstraea} correctly compute the earliest-arrival path from a source vertex $x$ to every vertex $v \in V$ within a given interval $[t_{\alpha}, t_{\omega}]$ with the complexity of $\mathcal{O}((\varepsilon_{s} + \varepsilon_{d} \cdot t_{max}) \cdot \log{}(|V|))$
\end{theorem}

\begin{theorem}
\label{theo:staticea}
When $\varepsilon_{s} \gg \varepsilon_{d}$, the complexity of Dijkstra-based algorithm become $\mathcal{O}(\varepsilon_{s} \cdot \log{}(|V|)$ while complexity of Wu's algorithm become $\mathcal{O}(\varepsilon_{s} \cdot t_{max})$
\end{theorem}

\begin{algorithm}[H]
 \caption{Dijkstra-based algorithm for Computing Earliest-Arrival Time}
 \label{alg:dijkstraea}
 \begin{algorithmic}[1]
 \renewcommand{\algorithmicrequire}{\textbf{Input:}}
 \REQUIRE Temporal Graph $G$, source nodes $S$, time interval $[t_{\alpha}, t_{\omega}]$\\
 \renewcommand{\algorithmicrequire}{\textbf{Output:}}
 \REQUIRE The earliest-arrival time from every source nodes $s \in S$ to every vertex  \\
 \STATE $PQ = Priority\_Queue$                                                                          \\
 \STATE $INSERT_{PQ}(t_{i}, s), \forall s \in S$                                                        \\
 \STATE $seen[s] = t_{i}, \forall s \in S$                                                              \\
 \STATE $arrvl\_time \leftarrow empty$ $dictionary$                                                   \\
 \STATE \textbf{while} $PQ \neq \emptyset$ \textbf{do}                                                  \\
 \STATE \quad $(t_u, u) \leftarrow POP\_MIN_{PQ}()$                                                 \\
 \STATE \quad \textbf{if} $u$ in $arrvl\_time$ \textbf{do}                                            \\
 \STATE \quad \quad continue                                                                            \\
 \STATE \quad $arrvl\_time[u] = t_u$                                                                  \\
 \STATE \quad \textbf{for} successor $v$ of $u$ \textbf{do}                                       \\
 \STATE \quad \quad \textbf{if} $(u, v)$ is a $static$ $edge$ \textbf{then}                         \\
 \STATE \quad \quad \quad $v\_arrvl \leftarrow t_u + 1$                                               \\
 \STATE \quad \quad \textbf{else} \textbf{then}                                                          \\
 \STATE \quad \quad \quad $v\_arrvl \leftarrow min\{t : t \in time\_labels(u, v),$ $and$ $t > t_u \}$
 \STATE \quad \quad \textbf{if} $v$ in $arrvl\_time$ \textbf{then}
 \STATE \quad \quad \quad \textbf{continue}
 \STATE \quad \quad \textbf{elseif} $v$ not in $seen$ or $v\_arrvl < seen[v]$ \textbf{do}
 \STATE \quad \quad \quad $seen[v] \leftarrow v\_arrival$
 \STATE \quad \quad \quad $INSERT_{PQ}(v\_arrvl, v)$
 \STATE \textbf{return} $arrvl\_time$
 \end{algorithmic}
\end{algorithm}

\subsection{EDO Algorithm for max-$RT$}

In this section, we discuss the application of the Evolutionary Diversity Optimization (EDO) algorithm within our problem context. The pioneering work of Goel et al. \cite{goel2022defending, goel2023evolving} introduced the EDO technique for addressing the edge-blocking problem in AD attack graphs. We initially applied Goel's EDO algorithm to our scenario. In our problem, the defender employs EDO to acquire a diverse set of defensive plans denoted as $C$, where the fitness function $f(C)$ can be obtained by computing the optimal attack plan. Let's define $P$ as the population of defensive solutions. An individual $p\in P$ is defined as the binarization of solution $C$ where each individual has a length of $|N_b|$, with 1 signifying the decision to block the corresponding node and 0 implying no blocking.

We initiate the process by generating a random population \(P\) of defensive solutions. An individual $p$ is randomly selected from $P$ to undergo either mutation or crossover, each with a probability of 0.5. The number $x$ of mutated bits in the offspring is chosen randomly based on a Poisson distribution. For \textbf{mutation}, we randomly select an individual $p'$ from $P$ and flip $x$ random bits, changing 0s to 1s and 1s to 0s. For example, if we choose $p' = \langle 1, 0, 1, 1, 0, 1 \rangle$ from $P$ and $x = 2$, the resulting offspring could be $p = \langle 0, 1, 0, 1, 1, 1 \rangle$. For \textbf{crossover}, we again randomly select two parents $p'$ and $p''$ from $P$. We identify $x$ coordinates where $p'$ has 0s and $p''$ has 1s, and flip the bits at those coordinates on both $p'$ and $p''$. Similarly, we identify $x$ coordinates where $p'$ has 1s and $p''$ has 0s, and flip the bits at those coordinates. After having the offspring using mutation and crossover operation, we add the new offspring to the population only if their fitness score is close to the best fitness score of the population and reject the individuals that contribute the least to the diversity of the population. We follow the \textbf{diversity measure} of population implementation of \cite{goel2022defending,goel2023evolving}. But to summarise our diversity measure aims to maximise the diversity of "unique" nodes in the population. Let $Cnt_P(v_i)$  be the function that counts the number of individuals in population $P$ that contain $v_i$. We say that $v_i$ is more "unique" to the population if they have a lower $Cnt_P(v_i)$ score. \textit{Again, we noted that this paper is not intended to redesign mutation, crossover operations, or diversity measures. Instead, our focus lies in the design of an algorithm aimed at enhancing the overall runtime and the convergence time to feasible solutions.}

Our preliminary investigation of the EDO algorithm revealed that most of generated offspring solution are infeasible. This challenge arises due to the expansive nature of the defender solution space (${|N_b| \choose b}$ combinations), which makes it difficult to generate feasible solutions using conventional evolution operators alone. 
Another challenge with the vanilla EDO algorithm is its requirement to execute the full fitness function. Although we have presented that the fitness function can be computed in polynomial time and pushed the runtime frontier by proposing a modification of Dijkstra-based for computing earliest-arrival paths. The algorithm execution time remains slow for larger graphs. In the following section, we will explore two constraint-handling techniques that we propose to enhance convergence to feasible solutions and improve the algorithm's runtime efficiency.

\subsection{Constraint-Handling Evolutionary Algorithm}
In this section, we shall introduce two constraint-handling approaches for our EDO algorithm.
\subsubsection{Integer Programming repair operator}
In this proposal, we aim to address the issue of infeasible offspring directly by employing a problem-specific \textit{repair} operator. The repair mechanism involves solving an Integer-Programming (IP) to "patch" the cutting solution. The complete algorithm can be describe as following. Suppose we encounter an infeasible offspring, denoted as $p$ after the mutation or crossover. For each blocked node, excluding those that have undergone a state change during the mutation or crossover, we probabilistically unblock them (i.e., change 1s to 0s) with a probability of 1/2. The purpose of this unblocking operation is to reserve additional space for the subsequent repair process and fulfill the cardinality condition. Finally, we solve our problem-specific IP repair operator. \textit{Due to space constraints, we will provide the detailed ILP formulation in the extended technical report, supplementing the main manuscript}. The ILP is formulated on the idea that if any node $i$ is connected to $j$ via edge $(i, j, t)$, and if $j$ can reach DA at any time before $t$, then $i$ can also reach DA at every time after $t+1$.

While the repair operator ensures convergence to a feasible solution in each iteration, it is worth noting that this approach is very memory costly. The IP requires $\mathcal{O}(|V|\cdot t_{max})$ variable and upto $\mathcal{O}(\varepsilon \cdot t_{max} + |V|)$ constraints, which can become exponentially large for certain graphs. Additionally, solving the IP itself is known to be a $\mathcal{NP}$-hard problem. 

\subsubsection{Surrogate-assisted and penalty-based repair operator}

Throughout our experiment, we observed that employing the full fitness function on the entire graph in each iteration proves to be very costly, especially for large graphs. Additionally, when using the vanilla EDO algorithm, we encountered difficulties as the solution failed to converge towards feasibility. 

To tackle the challenges mentioned earlier, we propose Algorithm \ref{alg:edo}. The core concept behind Algorithm \ref{alg:edo} is to evaluate the population on a lightweight surrogate fitness function in every iteration instead of the inefficient complete fitness function (\ref{equa:fitness}). The complete fitness function required to run Algorithm \ref{alg:optatk} on the whole graph. Our idea for design is that we only need to focus on a set of "important" paths that are likely to have the most impact on the evaluation, instead of inefficiently spending time on the entire graph. We will have two separate sets of populations in our algorithm namely global population $P_{global}$ and local population $P_{local}$. The local population is evaluated every iteration by the surrogate fitness, while the global population is only evaluated by the complete fitness function when a specific condition is met. Let $\Phi$ be the set of "important" temporal $(s, DA)$-path for the surrogate function. We initialize the set $\Phi$ by adding a random set of temporal paths in graph. Then we iteratively improve the function by adding to $\Phi$ the most up-to-date optimal attack path by the attacker when facing the current population. Our experimental results demonstrate that the surrogate function eventually becomes as effective as the complete fitness function. The proposed algorithm is designed to guide the solution towards convergence of the feasible solution. The pseudocode of the algorithm is presented in Algorithm \ref{alg:edo}. It involved the call of 3 other subroutines: 

\textit{Local Search (line 5):} In the local search, the algorithm performs the standard mutation or crossover, diversity measure and rejection. The key difference is that instead of using a resource-intensive fitness function, we employ a lightweight surrogate fitness function for evaluation. We say an individual $p$ is a \textbf{locally feasible} solution if $p$ can intercept all paths in $\Phi$. Individuals failed to block any paths in $\Phi$ will be penalized. The penalty score is determined by the number of paths in $\Phi$ that an individual $p$ cannot block. The \textbf{surrogate fitness function} can be presented as follows:

\begin{equation}
  f_s^{\phi}(C)=\begin{cases}
    \min_{\pi \in \Phi} RT(\pi, C), & \text{if $C$ is locally feasible}.\\
    -|\{ \pi \in \Phi : \pi \cap C = \emptyset\}|, & \text{otherwise}.
  \end{cases}
\end{equation}

\textit{Global Search (line 7):}  We define that global search starts only when there are no locally infeasible individuals in the local population, and a specified number of local iterations have been completed. In the Global Search, the algorithm adds each "candidate" individual from the local population to the global population and employs diversity measures and rejection on the global population. We use the complete fitness function to evaluate each individual. \textit{It's important to note that a solution $C$ is locally feasible may not necessarily be globally feasible}. This concern arises because the local search evaluates only a fraction of the graph ($\Phi$), which may not provide enough samples to form a cut in the graph. However, as stated in Theorem \ref{theo:converge}, we establish that eventually, the locally feasible solution yields the globally feasible solution after a certain number of iterations.

\textit{Update the Surrogate Fitness Function (line 8 - 12):} Following every global search, we improve the surrogate function by updating the important path set $\Phi$. The update is based on the performance of each individual in the local population. For every $p\in P_{local}$ that is globally infeasible, we add some random temporal (s, DA)-path to $\Psi$ in graph $G'=(V\setminus p, E)$ after removing nodes in cut set $p$. Those are the paths that make the individual $p$ globally infeasible. We use the modification of the Depth First Search algorithm for temporal graphs to find the random paths. For every $p\in P_{local}$ that is globally feasible, we improve the surrogate function by adding the optimal attack path when facing the defense solution $p$ to $\Phi$.

\begin{theorem}
\label{theo:converge}
In Algorithm \ref{alg:edo}, the number of iterations of Global Search until feasible solution $C$ on local evaluation function $f_{s}^{\phi}(C)$ yield feasible solution on the global evaluation function is $\mathcal{O}(|V|)$ iterations at worst. 
\end{theorem}


\begin{algorithm}[H]
 \caption{EDO with surrogate-assisted/penalty-based fitness function}
 \label{alg:edo}
 \begin{algorithmic}[1]
 \renewcommand{\algorithmicrequire}{\textbf{Input:}}
 \REQUIRE Temporal Graph $G$, honeypot budget $b$
 \renewcommand{\algorithmicrequire}{\textbf{Output:}}
 \REQUIRE Blocking population $P$  \\
 \STATE Initialize local population $P_{local}$ \\
 \STATE Initialize global population $P_{global}$
 \STATE Initialise set of paths $\Phi$
 \STATE \textbf{while} A termination criterion is met \textbf{do}
 \STATE \quad  $P_{local}$ $\leftarrow$ $LocalSearch(P_{local}, \phi)$
 \STATE \quad \textbf{if} $\prod_{p\in P_{local}, \pi \in \Phi} |p \cup \pi| \neq 0$ and global criterion is met \textbf{do}
 \STATE \quad \quad $P_{global}$ $\leftarrow$ $GlobalSearch(P_{global}, P_{local})$
 \STATE \quad \quad \textbf{foreach} $p \in P_{local}$ \textbf{do}
 \STATE \quad \quad \quad \textbf{if} $p$ is a $(S, DA)-cut$ in $G$ \textbf{do}
 \STATE \quad \quad \quad \quad Compute $\pi_{opt}^{p}(G)$ and add to $\Phi$
 \STATE \quad \quad \quad \textbf{else}
 \STATE \quad \quad \quad \quad Add a random paths from $s\in S$ to $DA$ in graph $G'=(V \textbackslash p, E)$ to $\Phi$
 \STATE \textbf{return} $P_{global}$
 \end{algorithmic}
\end{algorithm}

\section{Experiment Result}
\label{sec:exp}

\subsection{Experiment Set Up}
All of the experiments are carried out on a high-performance computing cluster with 1 CPU and 24GB of RAM allocated to each trial. As the real-world AD graph is sensitive, we will conduct experiments on synthetic graph generated by DBCreator \footnote{https://github.com/BloodHoundAD/BloodHound-Tools/tree/master/DBCreator} and Adsimulator \footnote{https://github.com/nicolas-carolo/adsimulator} - two state of the art tools for creating AD graphs. Every graph starting with R ("Rxxx") is generated by DBCreator while the one starting with label AD ("ADxxx") is generated by the ADsimulator. 
DBCreator only allows us to fine-tune the number of computers and users. In contrast, Adsimulator provides greater flexibility by enabling adjustments to various entities in the AD graph, including Security Groups, Organizational Units (OUs), Group Policy Objects (GPOs), and more. Consequently, we have two types of graphs generated by Adsimulator: 'ADX\textbf{x}', where default parameters are increased by a factor of '\textbf{x}' (e.g., ADX10 is 10 times the default setting), and 'ADU\textbf{y}', mimicking '\textbf{y}' fractional proportions of the structure of the real AD network at the University of Anonymized (e.g., ADU05 represents 5 $\%$ of the mimicked network). Due to space constraints, detailed information about the size of each graph will be provided in the technical report. However, for a quick estimate, here are the sizes of the largest graph for each type: R4000 (12001 nodes and 45780 edges), ADX20 (6013 nodes and 26671 edges), and ADU (6875 nodes and 37292 edges).

\begin{table}[h]
\begin{center}
\caption{Comparison of all algorithms with DBCreator's graph. The results show the average response time (higher is better) and the average last improvement time (lower is better) of each setting. The numbers in the parenthesis are the average last improvement time.}\label{tab:dbresult}
\smallskip\noindent
\resizebox{\linewidth}{!}{%
\begin{tabular}{lllll}
\hline
               & \textbf{R2000+C}   & \textbf{R4000+C} & \textbf{R2000+L} & \textbf{R4000+L} \\
\hline
\textbf{VAN-V}     & 2.10 (53177s) & 3.60 (42445s) & 0              & 0                \\
\textbf{VAN-D}     & 2.03 (68268s) & 4.09 (28514s) & 0              & 0           \\
\textbf{ILP-V}     & 4.07 (16715s) & 4.49 (17037s) & 2.62 (40999s) & 3.18 (41826s)              \\
\textbf{ILP-D}     & 4.09 (25177s) & 4.58 (19067s) & 2.90 (43272s) & 2.90 (33305s)               \\
\textbf{EST-V}     & 4.17 (302s)   & 4.70 (706s)   & 4.50 (11862s)  & 3.70 (20536s)              \\
\textbf{EST-D}     & 4.17 (473s)   & 4.70 (302s)   & 4.60 (10257s)  & 3.70 (20536s)              \\
\hline
\end{tabular}}
\end{center}
\end{table}

\begin{table*}[ht]
\begin{center}
\caption{Comparison all algorithms with ADsimulator's graph. No feasible result found by VIN so we did not include it here. OOM is stand for Out-of-Memory. All notion in Table \ref{tab:dbresult} will be also applied here. }\label{tab:compgreedy}
\smallskip\noindent
\resizebox{\linewidth}{!}{%
\begin{tabular}{lllllllll}
\hline
                                    & \textbf{ADX5+C}   & \textbf{ADX10+C} &\textbf{ ADX20+C} & \textbf{ADU5+C} & \textbf{ADX5+L}   & \textbf{ADX10+L} &\textbf{ ADX20+L} & \textbf{ADU5+L} \\
\hline
\textbf{ILP-V}                      & 3.21 (31362s)        & OOM               & OOM                   &  OOM       & 3.50 (26796s)        & 0.83 (78836s)              & OOM                   &  OOM                                 \\
\textbf{ILP-D}                      & 3.30 (25250s)         & OOM                & OOM                   & OOM    & 3.40 (25228s)         & 0.60 (73151s)                & OOM                   & OOM                                    \\

\textbf{EST-V}                      & 3.30 (10517s)        & 3.50 (20498s)        & 2.50 (34762s)            & 1.60 (67432s)     & 5.27 (1965s)        & 1.05 (38057s)        & 1.4 (68776s)            & 1.8 (70165s)                                      \\
\textbf{EST-D}                      & 3.40 (2415s)         & 3.30 (14839s)        & 2.50 (34365s)            & 1.70 (65532s)      & 4.77 (3175s)         & 1.75 (23709s)        & 1.90 (74284s)         & 1.40 (73482s)                                     \\
\hline

\end{tabular}}
\end{center}
\end{table*}

However, these tools only generate static snapshots of the graph. To generate a temporal AD attack graph, we will merge a "mould" of static AD graph with authentication data which simulates the characteristic of Hassession edge. The first source is the authentication data from The Comprehensive, Multi-Source Cyber-Security Events dataset [7], referred to as \textbf{LANL}. The second source is from an anonymous organization, labelled as \textbf{COMP}. We will provide the details of each dataset in the appendix. Combining these datasets involved the following process. First, in each static mould AD graph, we removed all HasSession edges. Next, we randomly mapped users, computers and authentication events from the authentication data to the mould graph to create an instance of the temporal graph. For clarity in denoting the generated instances, we referred to a temporal graph in the format $\{graph\}+\{auth\_source\}$. For instance, \textbf{R2000+C} indicates a temporal graph derived from the mould static graph \textbf{R2000}, with Hassession edge data sourced from the \textbf{C}OMP authentication dataset. In this notation, L refers to the LANL dataset, and C refers to the COMP dataset.

We use Gurobi 9.0.2 solver for solving the ILP module. For each experiment instance, we ran 10 trials. In each trial, we randomly choose a set of 10 starting nodes and randomly re-map the authentication data to the mould graph. 
To define the defensive budget for our problem, we have to determine the size of the minimum temporal cut $|minC|$. We will discuss how we determine $minC$ in our appendix. Given that the condition $b \geq |minC|$ has to be met to ensure our problem is feasible, we define the budget for our problem as $b = b_f*|minC|$ where $b_f > 1$ is the budget factor. We set the budget factor to $b_f = 1.5$ for every experiment. We define that only 90 percent of nodes in the graph is blockable. To construct the HasSession edges, we captured snapshots of the authentication dataset every 1 hour. In the experiment, we considered a total of 1000 snapshots in each setting (about 40 days). To avoid confusion in metrics, we will use "time unit" as the metric for the response time. We generate a population of 10 defensive blocking plans. The termination condition for the evolution algorithms was set at 2,000,000 iterations or 24 hours, whichever came first.

In our experiment, we adopt specific denotation for clarity: the Integer Linear Programming approach is denoted as ILP, the surrogate-assisted approach as EST, and the vanilla EDO algorithm as VIN. Additionally, we introduce a \textbf{Value-based Evolutionary Computation (VEC)} which greedily rejects the worst individual from the population instead of rejecting individuals based on diversity measure. In total, we will have 6 sets of algorithm includes: Vanilla EDO algorithm (\textbf{VAN-D}), Vanilla VEC algorithm (\textbf{VAN-V}), ILP-repair approach with EDO framework (\textbf{ILP-D}), ILP-repair approach with VEC framework (\textbf{ILP-V}), Surrogate-assisted approach with EDO framework (\textbf{EST-D}), Surrogate-assisted approach with VEC framework (\textbf{EST-V}). 
Note that our vanilla EDO algorithm's fitness function is implemented with our Dijkstra-based algorithm for computing the earliest-arrival path. Results would significantly degrade if Wu's algorithm were employed.
 


\subsection{Result Interpretation}

\begin{figure}[h]
  \includegraphics[width=1\linewidth]{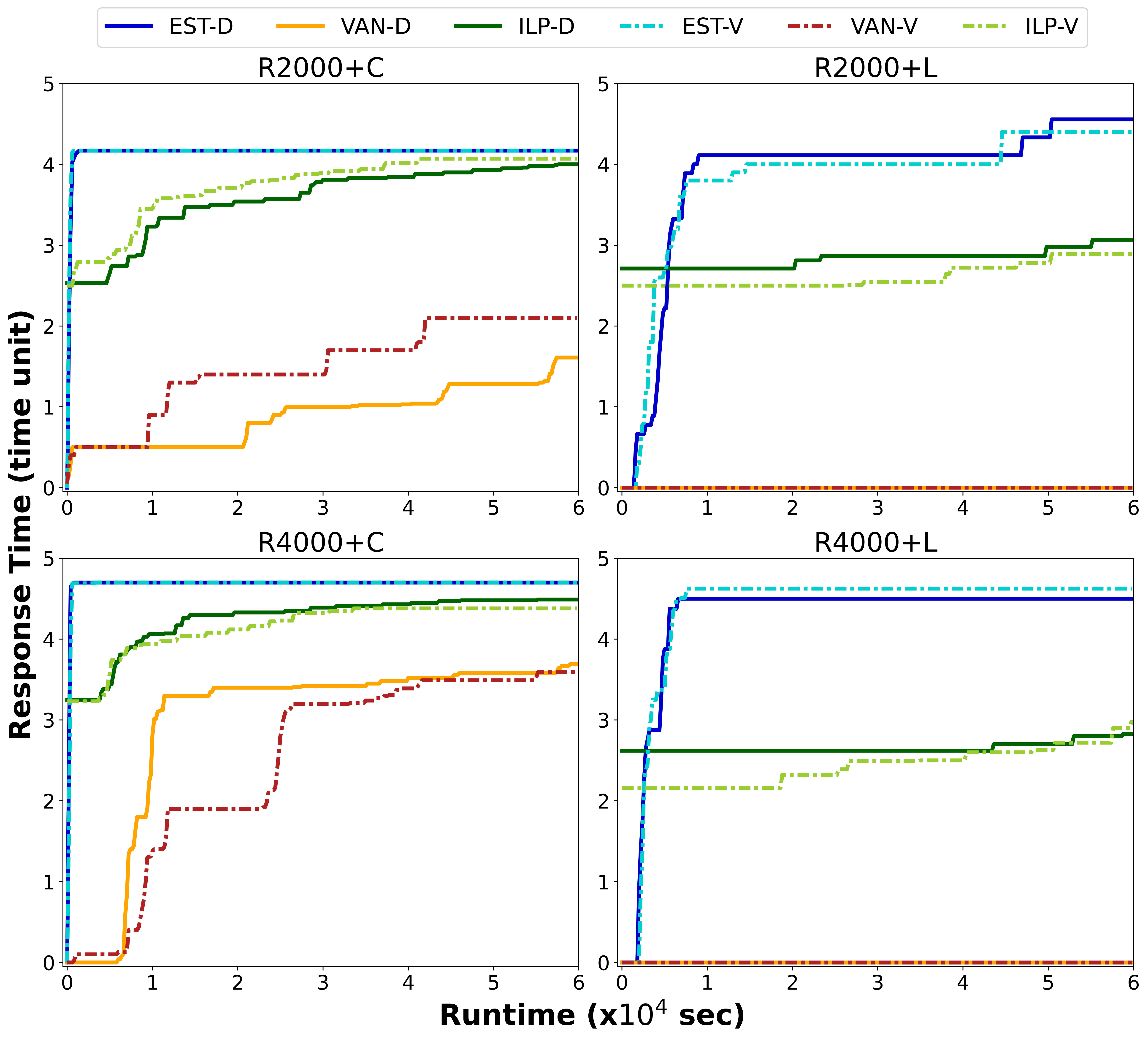}
  \caption{Performance comparison of all 6 algorithms. The \textbf{EST} approaches exhibit significantly faster convergence to the best result compared to the other methods.}
  \label{fig:converge}
\end{figure}

From Figure \ref{fig:converge}, the EST approach significantly improves the convergence speed of the Evolutionary Algorithm, allowing it to reach the best result much faster than ILP and VAN. Notably, ILP approach can find feasible solution from the early iteration since the repair operator will guarantee the mutation/crossover yields a feasible defensive solution. However, solving Integer Linear Programming itself is highly resource-intensive and is the bottleneck for this technique. Unfortunately, ILP failed to run in 5 out of 12 graphs due to Out-of-Memory errors.

Among the graphs, R2000+C and R4000+C are the only two where VAN can find any feasible solution. When we record the time to find the feasible solution (R2000+C and R4000+C), while vanilla takes about 21,728 seconds to find the feasible solution, EST takes on average 201 seconds which is about 108 times faster. For our setting, the EST method performs, on average, about $23\%$ better than the ILP. The convergence speed of EST is also superior to ILP.

To compare the performance of EDO-based algorithms (ended with D) with VEC-based algorithms (ended with V), we conducted a head-to-head comparison between these two approaches. Out of 21 comparable settings (excluding those with OOM errors and infeasible solutions), EDO outperformed VEC in 12 settings, while VEC performed better in only 9 cases (in instances where the response times were equal, we compared the average last improvement time). Overall, EDO outperformed VEC when applied to our problem. 

\section{Conclusion}

This paper investigated a Stackelberg game model between an attacker and a defender in temporal Active Directory attack graphs. We propose the use of Evolutionary Diversity Optimization algorithms to address this problem. However, the vanilla EDO encounters challenges when scaling to larger graphs and struggling to find feasible solutions.
To improve our solution, we first improve the computation of the attacker's optimal path (fitness function) by refining the calculation of the earliest-arrival path. Our novel Dijkstra-based algorithm for computing the earliest-arrival path, based on the observation that a significant portion of the AD infrastructure remains static. Experimentally, our algorithm is approximately 5 times faster than the SOTA algorithm when running on AD-specific graphs. Next, we introduce two constraint-handling techniques: a repair mechanism using Integer Linear Program (ILP) and a surrogate-assisted model with a penalty fitness function (EST). 
While ILP guarantees to find a feasible solution in early iterations, the EST method achieves this approximately 108 times faster than the vanilla approach. Moreover, EST outperforms ILP, demonstrating approximately a 23\% improvement in our specific setting.

\begin{acks}
\end{acks}

\bibliographystyle{ACM-Reference-Format}
\bibliography{main}
\appendix

\section{Appendix}

\subsection{Proof for Theorem \ref{theorem:np}}
\begin{proof} The proof is based on a reduction from the strict temporal $(s, d)$-seperator (strict-TS) problem which is $\mathcal{NP}$-complete \cite{zschoche2020complexity} for graph of lifetime $\geq 5$. 

\textbf{PROBLEM: } Strict-TS
\begin{itemize}
    \item \textbf{Input:} A temporal graph $G = (V, E_1,\cdots, E_{t_{max}})$, source node $s\in V$, destination $d\in V$ and $k \in \mathbb{N}$
    \item \textbf{Question:} Does $G$ admit a temporal $(s, d)$-seperator of size at most $k$
\end{itemize}
The proof gadget for the strict-TS is illustrated in Figure \ref{fig:NPgadget}.a and the complete proof is provided in \cite{zschoche2020complexity}. 

The high level idea for the hardness proof of max-$RT$ that the solution for the max-$RT$ problem can be found via solving the strict-TS problem. Let us define an instance of strictTS problem $G_{ts} = (V_{ts},  E_{ts,1},\cdots, E_{ts, t_{max}})$. We define a source node $s \in V_{ts}$ and destination node $d \in V_{ts}$. For the detailed construction of other nodes and edges in strict-TS, we refer the reader to Theorem 3.1 of \cite{zschoche2020complexity}. Let $minC_{ts}$ represent the solution to the strict-TS problem.

Subsequently, we construct the proof gadget for the max$RT$ problem (Figure \ref{theorem:np}.b) as follows. We introduce two entry nodes, $s_1$ and $s_2$. At time $t_\alpha$, node $s_1$ is connected to node $s$ of a sub-graph constructed following the strict-TS instance. At time $t_{\alpha}+6$, we also connect $d$ from the strict-TS subgraph to $y_2$. We delay every edges in the Strict-TS instance by $t_{\alpha}$. We assume that $s$ and $d$ is not blockable. Assuming $s$ and $d$ are not blockable, we finalize the instance by adding the following remaining edges: $(y_2, DA, t_{\alpha} + 7)$, $(s_2, y_1, t_{\alpha})$, and $(y_1, DA, t_{\omega})$ where $t_{\omega} \geq t_{\alpha}+7$. The full construction for max$RT$ can be seen in Figure \ref{theorem:np}.b.

With a defensive budget of $b = |minC_{ts}| + 1$, the optimal allocation involves locating the solution for the strict-TS instance and blocking vertices $y_2$. As the optimal solution of max$RT$ yield the optimal solution for Strict-TS, this implies that max$RT$ is $\mathcal{NP}$-hard.

\begin{figure}[h]
  \includegraphics[width=1\linewidth]{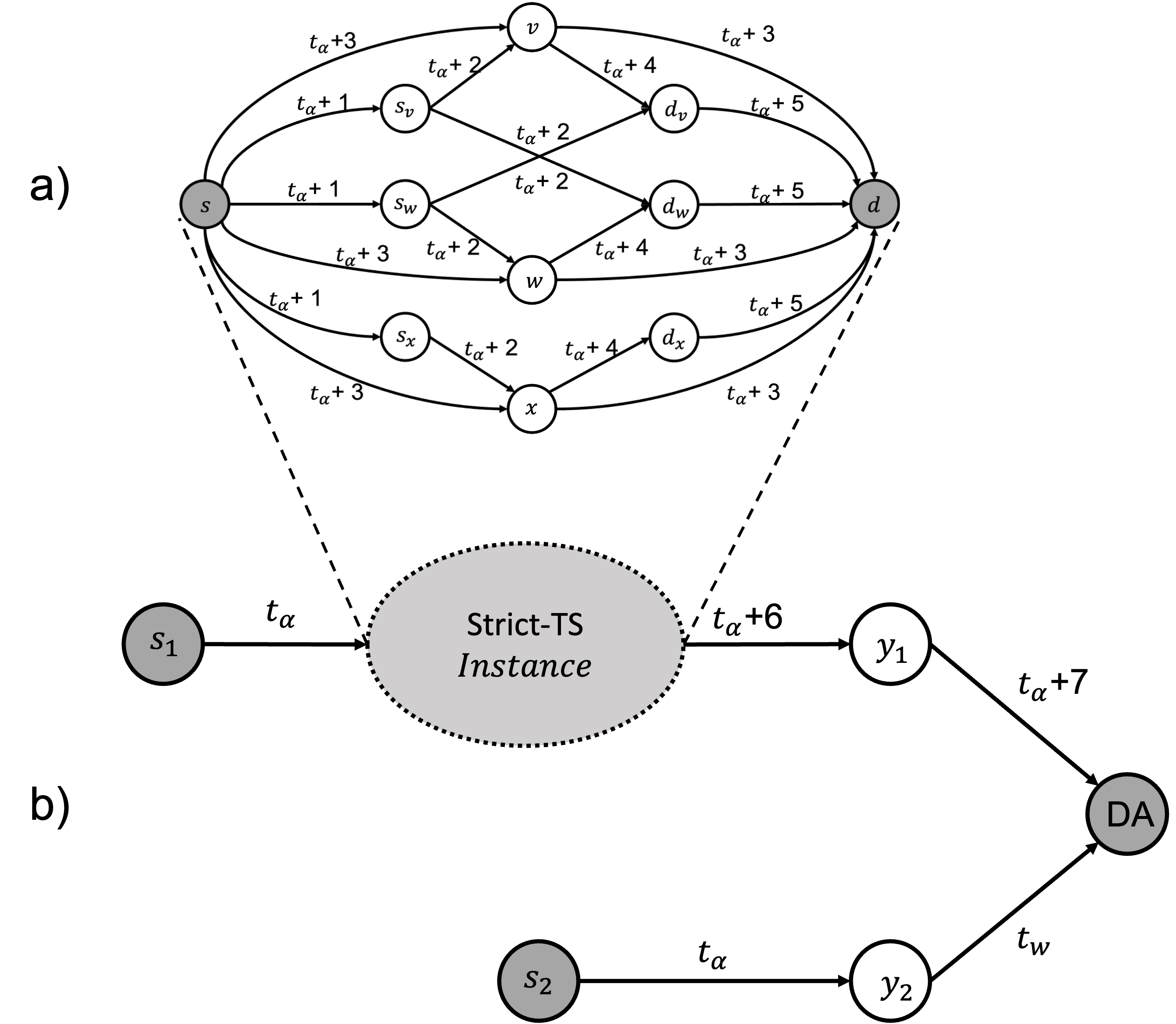}
  \caption{Proof gadget for Theorem \ref{theorem:np}. a) Proof gadget for Strict-TS problem. b) Proof gadget for max-$RT$ problem}
  \label{fig:NPgadget}
\end{figure}
\end{proof}

\subsection{Complete formulation for ILP repair operator}

We begin by introducing the key \textbf{variables}. Let $R_{i, t}$ be a binary variable representing the DA-reachability of node $i$. A value of 1 indicates that we can reach the DA from node $i$ when starting the journey at time $t$, while 0 indicates otherwise. Additionally, we define $B_i$ as a binary decision variable; a value of 1 mean we decide to block node $i$, and 0 otherwise. 

The \textbf{objective function} of the repair process is formulated as follows: $\min \sum\limits_{s\in S}\sum\limits_{t = t_{\alpha}}^{t_{\omega}} R_{s,t}$. The IP minimise number of starting nodes that can reach DA. A resulting objective function score of 0 mean the IP have successfully patch of the solution. Conversely, if the objective function score is greater than 0, it indicates the infeasibility of patching the cutting solution.

The IP is subject to various \textbf{constraints}. Firstly, for all $(u, v, t) \in E$ where $v\in N_b\setminus V$, we impose the constraint $R_{u, t} \geq R_{v, t+1}$. This constraint implies that if node $v$ can reach the DA when starting to traverse at time $t+1$, then we can reach the DA from $u$ when starting to traverse at time $t$. The "$\geq$" sign, rather than "$=$," accommodates cases where there is an alternate edge from $u$ to reach the DA.

Similarly, for all $v \in V$ and $t \in [t_{\alpha}, t_{\omega}]$, we have the constraint $R_{u, t} \geq R_{u, t+1}$. This indicates that if node $u$ can reach the DA when departing from this node at time $t+1$, then we can also reach the DA when departing from this node at time $t$.

Next, the blocking constraint is expressed as follows: for all $(u, v, t) \in E$ where $v\in N_b$, the constraint is $R_{u, t} \geq R_{v, t+1} - B_v$. This states that if node $v$ is decided to be blocked, then $u$ cannot reach the DA via the edge $(u, v)$.

Finally, we incorporate budget constraints: $\sum_{i \in V}\ B_{i} \leq b$ to conclude the formulation.

The complete formulation is presented as following:
\begin{subequations}
\begin{align}
\text{min} \displaystyle \sum\limits_{s\in S}\sum\limits_{t = t_{\alpha}}^{t_{\omega}} R_{s,t}  \nonumber\\ & \nonumber\\
    R_{u, t} \geq R_{v, t+1}& - B_v, & \forall (u, v, t) \in E, v \in N_b \\
    R_{u, t} \geq R_{v, t+1}&, & \forall (u, v, t) \in E, v \in V \setminus N_b \\
    R_{u, t} \geq R_{u, t+1}&, & \forall v \in V, t \in [t_{\alpha}, t_{\omega}] \\
    \sum_{i \in V}\ B_{i} \leq b,& \\
    R_{u, t}, B_{i} \in \{0, 1\}&
\end{align}
\end{subequations}

\subsection{Proof of Theorem \ref{theo:dijkstra}}

\begin{proof} To proof the correcness, we first provide the following Lemma: 
\begin{lemma}
\label{lemma:subpath}
Let a node sequence $V(\pi) = \langle x, v_{1}, v_{2}, \cdots, v_{k}  \rangle$ be the earliest-arrival path from vertex $x$ to vertex $v_{k}$ within some interval $[t_{\alpha}, t_{\omega}]$. Every prefix-subpath $V(\hat{\pi}) = \langle x, v_{1}, v_{2}, \cdots, v_{i}  \rangle \subset \pi$ where $0 < i < k$, is also an earliest-arrival path from $x$ to $v_{i}$ within $[t_{\alpha}, t_{\omega}]$.
\end{lemma} 
\begin{proof} Admit proof from Lemma 6 of \cite{wu2014path}
\end{proof}

The classic Dijkstra's algorithm computing single-source shortest paths based on the observation that the prefix-subpath of the shortest path is also a shortest path. Lemma \ref{lemma:subpath} implied that the prefix-subpath of an earliest-arrival path is also an earliest-arrival paths. This proof the correctness of the use of Dijkstra greedy strategy for computing earliest-arrival path.

We assume the use of a Priority Queue to identify the minimum arrival time of unvisited nodes in the Dijkstra-based algorithm. The algorithm grow the earliest arrival path by scan through each out-bound underlying edges in underlying edge the from the current node. Eventually, vertices $v \in V$ will be added to the heap once, hence, the worst-case heap size is $|V|$. Consequently, the complexity of the extract-min operation of the priority queue is $\mathcal{O}(\log{}(|V|))$. Iteratively popping the minimum value from the priority queue takes $\mathcal{O}(|V| \cdot \log{}(|V|)$. Since each node is only extracted once and not revisited, the for loop at line 10 will visit each underlying edge $E_{\downarrow} \in G_{\downarrow}$ only once. The updated earliest arrival time for each successor requires $\mathcal{O}(1)$ for static edges $e_{s} \in E_{s}$ and $\mathcal{O}(t_{max})$ for dynamic edges $e_{s} \in E_{d}$ where $t_{max} = t_{\omega}-t_{\alpha}$. Consequently, the overall complexity of the algorithm is $\mathcal{O}(|V| \cdot \log{}(|V|) + (\varepsilon_{s} + \varepsilon_{d} \cdot t_{max}) \cdot \log{}(|V|))$. As $\varepsilon_{\downarrow} = V^2$ and $\varepsilon_{\downarrow} = \varepsilon_{s} + \varepsilon_{d}$, we can simply rewrite as $\mathcal{O}((\varepsilon_{s} + \varepsilon_{d} \cdot t_{max}) \cdot \log{}(|V|))$.
\end{proof} 

\subsection{Proof of theorem \ref{theo:staticea}}
\begin{proof} When $\varepsilon_{s} \gg \varepsilon_{d}$, we can safely assume that $\varepsilon_{d} \to 0$ to present the complexity in term of $\varepsilon_{s}$. The complexity of our Dijkstra-based algorithm can be reformulated as $\mathcal{O}(\lim_{\varepsilon_{d} \to 0}(\varepsilon_{s} + \varepsilon_{d} \cdot t_{max})\cdot \log{}(|V|)$, which simplifies to $\mathcal{O}(\varepsilon_{s} \cdot \log{}(|V|)$. Similarly,  the complexity of Wu's algorithm in the same limit is $\mathcal{O}(\lim_{\varepsilon_{d} \to 0}(\varepsilon_{s} \cdot t_{max} + \varepsilon_{d} \cdot t_{max})$, which simplifies to $\mathcal{O}(\varepsilon_{s} \cdot t_{max})$.
\end{proof}

\subsection{Supplement pseudocode for Algorithm \ref{alg:edo}}
\begin{algorithm}[H]
 \caption{EDO's Local Search}
 \label{alg:localsearch}
 \begin{algorithmic}[1]
 \renewcommand{\algorithmicrequire}{\textbf{Input:}}
 \REQUIRE Local population $P_{local}$, Evaluation path set $\phi$\\
 \renewcommand{\algorithmicrequire}{\textbf{Output:}}
 \REQUIRE Blocking population $P$  \\
 \STATE Randomly select one (or two) parent $p_1$ (or and $p_2$) from $P_{local}$
 \STATE Generate a new solution $p_{3}$ by either mutation or crossover.
 \STATE $P_{local}$ $\leftarrow$ $EDO\_reject_{local}(P_{local}, \phi, p_{3})$
 \STATE \textbf{return} $P_{local}$
 \end{algorithmic}
\end{algorithm}

\begin{algorithm}[H]
 \caption{EDO's Global Search}
 \label{alg:globalsearch}
 \begin{algorithmic}[1]
 \renewcommand{\algorithmicrequire}{\textbf{Input:}}
 \REQUIRE Global population $P_{global}$, Local population $P_{local}$
 \renewcommand{\algorithmicrequire}{\textbf{Output:}}
 \REQUIRE Blocking population $P$  \\
 \STATE \textbf{foreach} $p \in P_{local}$ \textbf{do}
 \STATE \quad $P_{global}$ $\leftarrow$ $EDO\_reject_{global}(P_{global}, p)$
 \STATE \textbf{return} $P_{global}$
 \end{algorithmic}
\end{algorithm}

\subsection{Proof of Theorem \ref{theo:converge}}

\begin{proof}
Let's us denote $\Pi(\pi)  = \{\hat{\pi} : V(\hat{\pi}) = V(\pi)\}$ is the set of path where each of the element $\hat{\pi}$ have the identical path sequence with $\pi$. We make a following observations regarding the first point of contact of $\pi$: Let's say $i\in V(\pi) \cap C$ is the first point of contact of temporal path $\pi$, then, $i$ is also the first point of contact of every temporal path $\hat{\pi} \in \Pi(\pi)$. Based on the above mentioned observation, for each time the algorithm execute line 12 to add a random path to $\phi$, the algorithm will add a temporal path that will not overlap with any node sequence of any path in $\phi$.

Let's consider an instance of temporal graph denoted as $G = (V, E)$. In this graph, we have source vertices $s\in V$ and destination vertices $d \in V$, forming the underlying graph $G_{\downarrow} = (V, E_{\downarrow})$. It is specified that $G_{\downarrow}$ contains $\mathcal{O}(|V|-2)$ (excluding the source and destination vertices) disjoint paths from $s$ to $d$. Additionally, it is assumed that there is a budget available for deploying at least $|V|-2$ honeypots. The number of budget is $|V|-2$ since it is the size of the minimal temporal cut of our instance. If $b < |V|-2$, the response time is 0, defining the best defense. To simplify our proof, we make the assumption that the algorithm adds only one path to the set $\phi$ in each global iteration. If more than one path is added to the set, the algorithm may achieve faster convergence. The algorithm continues to append new temporal paths to the set $\phi$ until no further paths remain. In the worst-case scenario, each path $\pi$ added to $\phi$ corresponds to vertices disjoint paths in the underlying graph $G_{\downarrow}$ (every paths in $phi$ are vertices disjoint with each other). Consequently, all $\mathcal{O}(|V|-2)$ paths must be incorporated into the surrogate path set $\phi$ until the Local Search's feasible solution produces a (s, d)-cut on the graph $G$, meeting the feasibility condition for Global Search. This leads to the conclusion that, at worst, we need $\mathcal{O}(|V|)$ Global Search iterations until the feasible solution of Local Search can yield a feasible solution for Global Search. It's worth noting that in the event of tie-breaking, where paths added to $\phi$ aren't disjoint, the algorithm converges faster. Blocking common vertices demands less budget, resulting in $\phi$ containing only disjoint paths as the worst-case scenario.

\end{proof}

\subsection{ILP for finding minimum temporal cut}
\cite{zschoche2020complexity} provide the complexity analysis on the minimum temporal cut problem (Strict-TS). Despite our effort in finding algorithm for Strict-TS in the literature, we have not come across any algorithm for this algorithm yet. Here, we proposed an ILP formulation to optimally solve the problem. The ILP formulation based on the idea that if node $u$ can reach the DA when departing from this node at time $t+1$, then we can also reach the DA when departing from this node at time $t$. The formulation is inspired by an ILP repair operator, with slight modifications to accommodate our problem requirements. We remove the budget constraints, and the objective function is tailored to minimize the number of budget allocations for the cut. The formulation is presented as follow:

\begin{subequations}
\begin{align}
\text{min} \displaystyle \sum\limits_{v \in N_b} B_{v}  \nonumber\\ & \nonumber\\
    R_{u, t} \geq R_{v, t+1}& - B_v, & \forall (u, v, t) \in E, v \in N_b \\
    R_{u, t} \geq R_{v, t+1}&, & \forall (u, v, t) \in E, v \in V \setminus N_b \\
    R_{u, t} \geq R_{u, t+1}&, & \forall v \in V, t \in [t_{\alpha}, t_{\omega}] \\
    R_{u, t}, B_{i} \in \{0, 1\}&
\end{align}
\end{subequations}
\end{document}